\documentclass[a4paper,11pt]{scrartcl}
\usepackage[margin=3cm]{geometry}
\usepackage{hyperref}
\usepackage{amssymb, amsmath, amsthm}
\usepackage{microtype}
\usepackage[inline]{enumitem}
\usepackage{gensymb}
\usepackage{amssymb}
\usepackage{amsmath}
\usepackage{graphicx}
\usepackage{wrapfig}
\usepackage{textcomp}
\usepackage{tabularx}
\usepackage{latexsym}
\usepackage{todonotes}
\usepackage{xcolor}
\usepackage{subfig}
\definecolor{darkblue}{rgb}{0,0,0.4}
\definecolor{darkred}{rgb}{0.5,0,0}
\hypersetup{colorlinks=true,linkcolor=darkblue,citecolor=darkred,urlcolor=blue}

\graphicspath{{figures/}} 

\theoremstyle{plain}
\newtheorem{theorem}{Theorem} 
\newtheorem{lemma}{Lemma}

\renewcommand{\vec}[1]{\overrightarrow{#1}}
\newcommand{\eps}{\varepsilon}
\renewcommand{\angle}{\measuredangle}

\title{Strongly Monotone Drawings of Planar Graphs\footnote{This research was 
	initiated during the Geometric Graphs Workshop Week (GGWeek'15) at the 
	FU Berlin in September 2015. Work by P. Kindermann was supported by DFG 
	grant SCHU2458/4-1. Work by M. Scheucher was partially supported by the ESF 
	EUROCORES programme EuroGIGA -- CRP ComPoSe, Austrian Science Fund (FWF): 
	I648-N18 and FWF project P23629-N18 `Combinatorial Problems on Geometric Graphs'.}}

\author{%
	Stefan~Felsner\thanks{Institut f\"ur Mathematik, 
  Technische Universit\"at Berlin, Germany}
	\and
	Alexander~Igamberdiev\thanks{LG Theoretische Informatik,
  FernUniversit\"at in Hagen, Germany}
	\and
	Philipp~Kindermann\footnotemark[3]
	\and
	Boris~Klemz\thanks{Institute of Computer Science, 
  Freie Universit\"at Berlin, Germany}
	\and
	Tamara~Mchedlidze\thanks{Institute of Theoretical Informatics, 
  Karlsruhe Institute of Technology, Germany}
	\and
	Manfred~Scheucher\thanks{Institute of Software Technology,
  Graz University of Technology, Austria}}

\date{}

\begin{document}
\maketitle

\begin{abstract}
A straight-line drawing of a graph is a \emph{monotone drawing} if for
each pair of vertices there is a path which is monotonically increasing in 
some direction, and it is called a \emph{strongly monotone drawing} if the
direction of monotonicity is given by the direction of the line segment
connecting the two vertices.

We present algorithms to compute crossing-free strongly monotone
drawings for some classes of planar graphs; namely, 3-connected planar
graphs, outerplanar graphs, and 2-trees. The drawings of 3-connected
planar graphs are based on primal-dual circle packings.
Our drawings of outerplanar graphs depend  on a new algorithm
that constructs strongly monotone drawings of trees which are also convex.
For irreducible trees, these drawings are strictly convex.
\end{abstract}

\section{Introduction}
To find a path between a source vertex and a target vertex is one of
the most important tasks when data are given by a graph, c.f.~Lee et
al.~\cite{lppfh-ttfgv-BELIV06}.  This task may serve as criterion for
rating the quality of a drawing of a graph. Consequently researchers
addressed the question of how to visualize a graph such that finding a
path between any pair of nodes is easy. A user study of Huang et
al.~\cite{heh-agrb-PVis09} showed that, in performing path-finding
tasks, the eyes follow edges that go in the direction of the target
vertex. This empirical study triggered the research topic of finding
drawings with presence of some kind of geodesic paths. Several
formalizations for the notion of geodesic paths have been proposed,
most notably the notion of strongly monotone paths.  Related drawing
requirements are studied under the titles of self-approaching drawings
and greedy drawings.

Let~$G=(V,E)$ be a graph. We say that a path~$P$ is \emph{monotone
with respect to} a direction (or vector)~$d$  if the
orthogonal projections of the vertices of~$P$ on a line with direction
$d$ appear in the same order as in~$P$.  A straight-line
drawing of~$G$ is called \emph{monotone} if for each pair of
vertices~$u,v\in V$ there is a connecting path that is monotone with
respect to some direction. To support the path-finding tasks
it is useful to restrict the monotone direction for each path
to the direction of the line segment
connecting the source and the target vertex: a
path~$v_1v_2\ldots v_k$ is called \emph{strongly monotone} if it is
monotone with respect to the vector~$\vec{v_1v_k}$.  A straight-line
drawing of~$G$ is called \emph{strongly monotone} if each pair of
vertices~$u,v\in V$ is connected by a strongly monotone path.

In this paper, we are interested in strongly  monotone drawings which are 
also planar.  If crossings are allowed, then any strongly  monotone
drawing of a spanning tree of $G$ yields a strongly monotone
drawing of $G$, this has been observed by Angelini et
al.~\cite{acbfp-mdg-12}.

\paragraph{Related Work.}

In addition to (strongly) monotone drawings, there are several other drawing
styles that support the path-finding task. The earliest studied is the concept
of \emph{greedy drawings}, introduced by Rao et al.~\cite{rpss-grwli-MOBICOM03}. 
In a greedy drawing, one can find a source--target path by iteratively
selecting a neighbor that is closer to the target.  Triangulations
admit crossing free greedy drawings~\cite{d-gdt-10}, and more
generally 3-connected planar graphs have greedy
drawings~\cite{lm-gems-10}.  Trees with a vertex of degree at least~6
have no greedy drawing. N\"ollenburg and Prutkin~\cite{np-egdt-ESA13}
gave a complete characterization of trees that admit a greedy
drawing.

Greedy drawings can have some undesirable properties, e.g., a greedy
path can look like a spiral around the target vertex. To get rid of
this effect, Alamdari et al.~\cite{acglp-sag-GD12} introduced a
subclass of greedy drawings, so-called \emph{self-approaching
  drawings} which require the existence of a source--target path such
that for any point $p$ on the path the distance to another point~$q$
is decreasing along the path.  In greedy drawings this is only
required for $q$ being the target-vertex.  These drawings are related
to the concept of self-approaching
curves~\cite{ikl-sac-MPCPS95}. Alamdari et al. provide a complete
characterization of trees that admit a self-approaching drawing.

Even more restricted are \emph{increasing-chord drawings}, which
require that there always is a source--target path which is
self-approaching in both directions.  N\"ollenburg et
al.~\cite{npr-osaic-arXiv14} proved that every triangulation has a
(not necessarily planar) increasing-chord drawing and every planar
3-tree admits a planar increasing-chord drawing. Dehkordi et
al.~\cite{dfg-icgps-GD14} studied the problem of connecting a given
point set in the plane with an increasing-chord graph.

Monotone drawings were introduced by Angelini et al.~\cite{acbfp-mdg-12} 
They showed that any $n$-vertex tree admits a monotone drawing on a grid of 
size $O(n^{1.6}) \times O(n^{1.6})$ or $O(n) \times O(n^2)$.  They also showed
that any 2-connected planar graph has a monotone drawing having exponential 
area. Kindermann et al.~\cite{kssw-omdt-GD14} improved the area bound to
$O(n^{1.5})\times O(n^{1.5})$ even with the property that the drawings
are convex. The area bound was further lowered to 
$O(n^{1.205})\times O(n^{1.205})$ by He and He~\cite{hh-cmdt-COCOON15}.
Hossain and Rahman~\cite{hr-mgdpg-FAW14} showed that every connected 
planar graph admits a monotone drawing on a grid of size $O(n) \times O(n^2)$. 
For 3-connected planar graphs, He and He~\cite{hh-md3pg-ESA15} proved
that the convex drawings on a grid of size $O(n) \times O(n)$,  produced by the 
algorithm of Felsner~\cite{f-cdpgo-O01}, 
 are monotone. For the fixed embedding setting, Angelini et 
al.~\cite{adkmrsw-mdgfe-13} showed that every plane graph admits a monotone
drawing with at most two bends per edge, and all 2-connected plane graphs and
all outerplane graphs admit a straight-line monotone drawing.

Angelini et al.~\cite{acbfp-mdg-12} also introduced the concept of
\emph{strong monotonicity} and gave an example of a drawing of a
planar triangulation that is not strongly monotone. Kindermann et 
al.~\cite{kssw-omdt-GD14} showed that every tree admits a strongly monotone
drawing. However, their drawing is not necessarily strictly convex and 
requires more than exponential area. Further, they presented an infinite
class of 1-connected graphs that do not admit strongly monotone drawings.
N\"ollenburg et al.~\cite{npr-osaic-arXiv14} have recently shown that 
exponential area is required for strongly monotone drawings
of trees and binary cacti.

There are some relations among the aforementioned drawing styles.
Plane increasing-chord drawings are self-approaching  by definition but also  strongly 
monotone. Self-approaching drawings are greedy by definition.
On the other hand, (plane) self-approaching drawings are not 
necessarily monotone, and vice-versa.

\paragraph{Our Contribution.} 
After giving some basic definitions used throughout the paper in
Section~\ref{sec:preliminaries}, we present four results.  First, we
show that any 3-connected planar graph admits a strongly monotone
drawing induced by primal-dual circle packings
(Section~\ref{sec:3connected}). Then, we answer in the affirmative the
open question of Kindermann et al.~\cite{kssw-omdt-GD14} on whether
every tree has a strongly monotone drawing which is strictly convex.
We use this result to show that every outerplanar graph admits a
strongly monotone drawing (Section~\ref{sec:trees-outerplanar}).
Finally, we prove that 2-trees can be drawn strongly monotone
(Section~\ref{sec:2trees}).  All our proofs are constructive and admit
efficient drawing algorithms. Our main open question is whether every
planar 2-connected graph admits a plane strongly monotone drawing
(Section~\ref{sec:conclusion}).  It would also be interesting to
understand which graphs admit strongly monotone drawings on a grid of
polynomial size.

\section{Definitions}\label{sec:preliminaries}
Let $G=(V,E)$ be a graph. 
A \emph{drawing}~$\Gamma$ of~$G$ maps the vertices of~$G$ to distinct points in 
the plane and the edges of~$G$ to simple Jordan curves between their end-points. 
A planar drawing induces a \emph{combinatorial embedding}
which is the class of topologically equivalent drawings. In particular, an 
embedding specifies the connected  regions of the plane, called \emph{faces}, whose boundary 
consists of a cyclic sequence of edges. The unbounded face is called the 
\emph{outer face}, the other faces are called \emph{internal faces}.
An embedding can also be defined by a \emph{rotation 
system}, that is, the circular order of the incident edges around a vertex. 
Note that both definitions are equivalent for planar graphs.

A drawing of a planar graph is a \emph{convex drawing} if it is
crossing free and internal faces are realized as convex
non-overlapping polygonal regions. The \emph{augmentation} of a drawn
tree is obtained by substituting each edge incident to a leaf by a ray
which is begins with the edge and extends across the leaf. A
drawing of a tree is a \emph{(strictly) convex drawing} if the augmented drawing
is crossing free and has (strictly) convex faces, i.e., all the angles
of the unbounded polygonal regions are less or equal to (strictly less than) $\pi$.  Note
that strict convexity forbids vertices of degree~2. We call a tree
\emph{irreducible} if it contains no vertices of degree~2.  It has been
observed before that a convex drawing of a tree is also monotone but a
monotone drawing is not necessarily convex,
see~\mbox{\cite{acbfp-mdg-12,acm-mpaoa-SoCG89}}.

A \emph{$k$-tree} is a graph which can be produced from a complete
graph~$K_{k+1}$ and then repeatedly adding vertices in such a way that 
the neighbors of the added vertex form a $k$-clique. We say 
that the new vertex is \emph{stacked}  on the clique.
By construction $k$-trees are chordal graphs. They can also be
characterized as maximal graphs with treewidth~$k$, that is, no edges can 
be added without increasing the treewidth. Note that $1$-trees are 
equivalent to trees and $2$-trees are equivalent to  maximal series-parallel
graphs.

We denote an undirected edge between two vertices~$a,b\in V$ by
$(a,b)$.  In a drawing of~$G$, we may identify each vertex with the
point in the plane it is mapped to.  For two vectors~$x$ and~$y$, we
define the angle~$\angle(x,y)$ as the smallest angle between the two
vectors, that is, $\angle(x,y) = \arccos \left( \frac{\langle
    x,y\rangle}{|x||y|}\right)$, and for three points $p,q,r$, we
define $\angle{pqr} = \angle(\vec{qp},\vec{qr})$.  We say that a
vector $x$ is monotone with respect to~$y$ if~$\angle(x,y)<\pi/2$.
This yields an alternative definition of a strongly monotone path: A
path $v_1v_2\dots v_k$ is strongly monotone if
$\angle(\vec{v_iv_{i+1}},\vec{v_1v_k}) <\pi/2$, for $1\le i\le k-1$. Note
that we interpret monotonicity as strict monotonicity, i.e., we do not
allow edges on the path that are orthogonal to the segment between the
endpoints.

\section{3-Connected Planar Graphs}\label{sec:3connected}
\def\CC{{\cal C}}
\def\RR{\mathbb{R}}

In this section, we prove the following theorem. 

\begin{theorem}
Every 3-connected planar graph has a strongly monotone drawing.
\end{theorem}

\begin{proof}
We show that the straight-line drawing corresponding to a primal-dual
circle packing of a graph $G$ is already strongly monotone.  The
theorem then follows from the fact that any 3-connected planar graph~$G=(V,E)$
admits a primal-dual circle packing. This was shown by Brightwell and
Scheinerman~\cite{brightwell:1993}; for a comprehensive
treatment of circle packings we refer to Stephenson's book~\cite{st-icp-05}.

A \emph{primal-dual circle packing} of a plane graph~$G$ consists of
two families $\CC_V$ and~$\CC_F$ of circles such that, there is a
bijection $v \leftrightarrow C_v$ between the set $V$ of vertices of
$G$ and circles of $\CC_V$ and a bijection $f \leftrightarrow C_f$
between the set $F$ of faces of $G$ and circles of $\CC_F$.  Moreover,
the following properties hold:
\begin{enumerate}[label=(\arabic*)]
	\item 
		The circles in the family $\CC_V$
		are interiorly disjoint and their contact graph is~$G$, i.e.,
		$C_u \cap C_v \neq \emptyset$ if and only if $(u,v) \in E(G)$.
	\item
		If $C_o\in\CC_F$ is the circle of the outer face~$o$, then 
		the circles of $\CC_F\setminus\{C_o\}$ are interiorly
		disjoint while $C_o$ contains all of them. The contact graph 
		of $\CC_F$ is the dual~$G^*$ of~$G$, 
		i.e., $C_f \cap C_g \neq \emptyset$ if and only if $(f,g) \in E(G^*)$.
	\item
		The circle packings~$C_V$ and~$C_F$ are orthogonal, i.e.,
		if~$e=(u,v)$ and the dual of~$e$ is~$e^*=(f,g)$, then there is a point $p_e =
		C_u \cap C_v = C_f \cap C_g$; moreover, the common tangents $t_{e^*}$ of 
		$C_u$, $C_v$ and $t_{e}$ of $C_f$, $C_g$ cross perpendicularly in $p_e$.
\end{enumerate}

\begin{figure}[tb]
  \centering
  \subfloat[\label{fig:primalDualRep}]{
    \centering
    \includegraphics{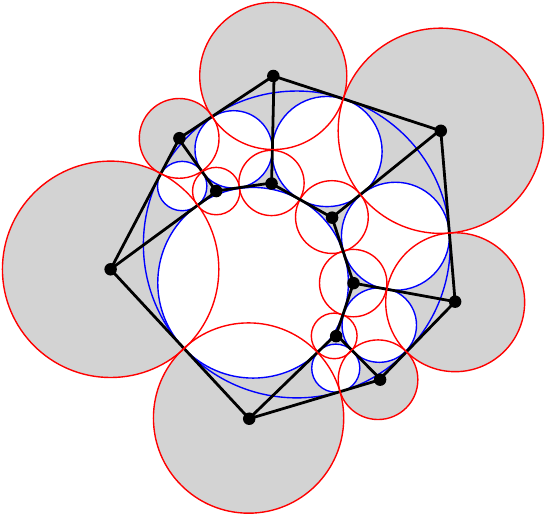}}
  \hfill
  \subfloat[\label{fig:3conMonPath}]{
    \centering
    \includegraphics{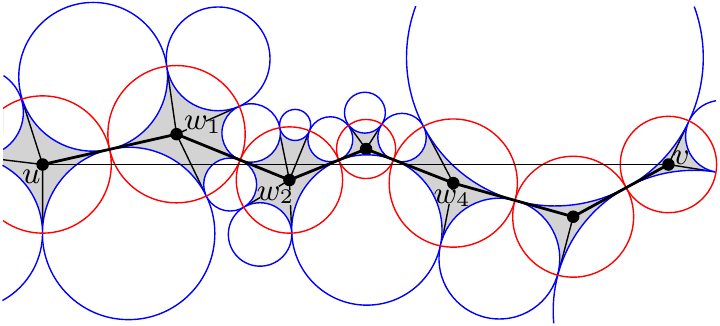}}
  \caption{(a) Drawing $\Gamma$ of 3-connected graph $G=(V,E)$. 
    Red circles are vertex circles $\mathcal C_V$, Blue circles are face circles $\mathcal C_F$. Regions of faces in white, regions of vertices in gray. (b) A strongly monotone path (thick edges) 
    from~$u$ to~$v$.}
  \label{fig:3con}
\end{figure}

\noindent
Let a primal-dual circle packing of a graph $G$ be given. For each vertex $v$,
let $p_v$ be the center of the corresponding circle $C_v$.  By placing each
vertex $v$ at $p_v$, we obtain a planar straight-line drawing~$\Gamma$
of~$G$. In this drawing, the edge $e=(u,v)$ is represented by the segment with
end-points $p_u$ and $p_v$ on $t_e$.  The face circles are inscribed circles
of the faces of~$\Gamma$; moreover, $C_f$ is touching each boundary edge of
the face~$f$; see Figure~\ref{fig:primalDualRep}. 

A straight-line drawing $\Gamma^*$ of the dual $G^*$ of $G$ with the dual
vertex of the outer face $o$ at infinity can be obtained similarly by placing the dual vertex
of each bounded face $f$ at the center of the corresponding circle $C_f$.
In this drawing, a dual edge $e^*=(f,o)$ is represented by the ray 
supported by $t_{e^*}$ that starts at $p_f$ and contains $p_e$.

In the following, we will make use of a specific partition~$\Pi$
of the plane. The regions of~$\Pi$ correspond to the
vertices and the faces of $G$. For a vertex or face~$x$, let
$D_x$ be the interior disk of~$C_x$.
\begin{itemize}
\item The region $R_f$ of a bounded
face $f$ is $D_f$.
\item The region $R_v$ of a vertex $v$ is obtained from the
disk $D_v$ by removing the intersections with the disks of bounded faces, i.e.,
$R_v = D_v \setminus \bigcup_{f\neq o} R_f = D_v \setminus
\bigcup_{f\neq o} D_f$;
see~Figure~\ref{fig:primalDualRep}.
\end{itemize}
To get a partition of the whole plane, we assign the complement of the already
defined regions to the outer face, i.e, $R_o = \RR^2 \setminus (\bigcup_{f\neq o} R_f \cup \bigcup_{v} R_v) =
\RR^2 \setminus (\bigcup_{f\neq o} D_f \cup \bigcup_{v} D_v).$

Note that the edge-points $p_e$ are part of the boundary of four regions of
$\Pi$ and if two regions of $\Pi$ share more than one point on the boundary,
then one of them is a vertex region~$R_v$, the other is a face-region $D_f$, and
$(v,f)$ is an incident pair of $G$.

We are now prepared to prove the strong monotonicity of~$\Gamma$. Consider two
vertices~$u$ and~$v$ and let $\ell$ be the line spanned by $p_u$ and $p_v$.
W.l.o.g., assume that~$\ell$ is horizontal and~$p_u$ lies left of $p_v$.  Let
$\ell_\text{s}$ be the directed segment from $p_u$ to $p_v$.
Since $p_u\in R_u$ and $p_v\in R_v$, the segment $\ell_\text{s}$ starts and ends
in these regions. In between, the segment will traverse some other regions of
$\Pi$. This is true unless $(u,v)$ is an edge of $G$ whence the strong
monotonicity for the pair is trivial. We assume non-degeneracy
in the following sense.

\noindent{\itshape Non-degeneracy:} The interior of the segment $\ell_\text{s}$
contains no vertex-point $p_w$, edge-point $p_e$, or face-point $p_f$.
\smallskip

\noindent
M\"obius transformations of the plane map circle packings to circle
packings. In fact the primal-dual circle packing of $G$ is unique up
to M\"obius transformation, see~\cite{st-icp-05}. Now any degenerate primal-dual
circle packing of $G$ can be mapped to a non-degenerate one by a
M\"obius transformation. This justifies the non-degeneracy
assumption. Later we will give a more direct handling of
degenerate situations.

Let $u=w_0,w_1, \ldots ,w_k=v$ be the sequence
of vertices whose region is intersected by $\ell_\text{s}$, in the order of
intersection from left to right; see~Figure~\ref{fig:3conMonPath} and let $p_i=p_{w_i}$.
We will construct a strongly monotone path $P$ from $p_u$ to $p_v$ in $\Gamma$
that contains $p_u=p_0,p_1, \ldots ,p_k=p_v$ in this order.
Let~$P_i$ be the subpath of $P$ from $p_{i-1}$ to $p_i$. 
Since $\ell_\text{s}$ may revisit a vertex-region, it is possible that
$p_{i-1}=p_i$; in this case we set $P_i = p_i$. Now suppose that
$p_{i-1}\neq p_i$. Non-degeneracy implies that 
the segment $\ell_\text{s}$ alternates between
vertex-regions and face-regions; hence, a unique
disk~$D_f$ is intersected by $\ell_\text{s}$ between the regions of
$w_{i-1}$ and $w_i$. It follows that $w_{i-1}$ and~$w_i$ are 
vertices on the boundary of $f$. The boundary of $f$ contains
two paths from $w_{i-1}$ to~$w_i$. In $\Gamma$, one of these two
paths from $p_{i-1}$ to $p_i$ is above $D_f$; we call it
the \emph{upper path}, the other one is below $D_f$, this is the \emph{lower path}.
If the center $p_f$ of $D_f$ lies below~$\ell$, we choose the upper path
from~$p_{i-1}$ to~$p_i$ as $P_i$; otherwise, we choose the lower path.

\begin{figure}[tb]
\centering
\includegraphics[scale=1]{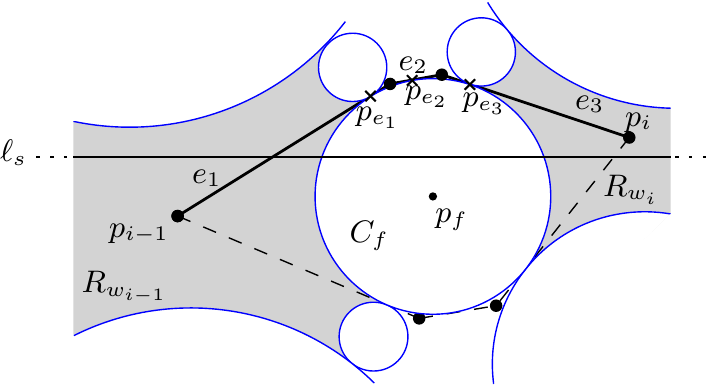}
\caption{The path~$P_i$ connecting $p_{i-1}$ and $p_i$.}
\label{fig:3conDetour}
\end{figure}

Suppose that this rule led to the choice of the upper path; see
Figure~\ref{fig:3conDetour}. The case that the lower path was chosen
works analogously. We have to show that $P_i$ is monotone with respect
to $\ell$, i.e., to the $x$-axis.  Let $e_1,\ldots,e_r$ be the edges
of this path and let $e_j=(q_{j-1},q_j)$; in particular $q_0=p_{i-1}$
and $q_r=p_i$. Since $R_{w_{i-1}}$ is star-shaped with center
$p_{i-1}$, the segment connecting $p_{i-1}$ with the first
intersection point of $\ell$ with $C_f$ belongs to
$R_{w_{i-1}}$. Therefore, the point $p_{e_1}$ of tangency of edge
$e_1$ at $C_f$ lies above $\ell$.  Similarly, $p_{e_r}$ and, hence,
all the points $p_{e_j}$ lie above $\ell$. Since the points
$p_{e_1},\ldots,p_{e_r}$ appear in this order on $C_f$ and the center
of $C_f$ lies below~$\ell$, we obtain that their $x$-coordinates are
increasing in this order.  This sequence is interleaved with the
$x$-coordinates of $q_0,q_1,\ldots,q_r$, whence this is also monotone.
This proves that the chosen path $P_i$ is monotone with respect to
$\ell$.  Monotonicity also holds for the concatenation
$P=P_1+P_2+\ldots +P_k$; see~Figure~\ref{fig:3conMonPath}.

We have shown strong monotonicity 
under the non-degeneracy assumption. Next we consider degenerate cases
and show how to find  strongly monotone paths in these cases.

If $\ell_\text{s}$ contains a vertex-point $p_w$ with $w\neq u,v$, the
path $P$ between $u$ and $v$ is just the concatenation of monotone
paths between the pairs $u,w$ and $w,v$; hence, it is strongly
monotone.  Next suppose that $\ell_\text{s}$ contains an edge-point
$p_e$.  If the edge $e$ in $\Gamma$ is horizontal, then we also have
two vertex-points on $\ell_\text{s}$ and are in the case described
above; otherwise, we consider the region which is touching $\ell$ from
above as intersecting and the region which is touching~$\ell$ from
below as non-intersecting.  This recovers the property that there is
an alternation between vertex-regions and face-regions intersected by
$\ell_\text{s}$. Hence, the definition of the path for $u$ and~$v$
gives a strongly monotone path unless it contains a vertical edge.
The use of a vertical edge can be excluded by properly adjusting
degeneracies of the form $p_f\in \ell$. For faces $f$ with~$p_f\in
\ell$, we use the upper path, i.e., we consider $p_f$ to be
below~$\ell$. Thus, even in degenerate situations the drawing
corresponding to a primal-dual circle packing is strongly
monotone. This concludes the proof.
\end{proof}

\section{Trees and Outerplanar Graphs}\label{sec:trees-outerplanar}

Kindermann et al.~\cite{kssw-omdt-GD14} have shown that any tree has a strongly 
monotone drawing and that any irreducible binary tree
has a strictly convex strongly monotone drawing. 
They left as an open question whether every tree admits a convex strongly 
monotone drawing; noticing that, in the positive case, this would imply that 
every Halin graph has a  convex strongly monotone drawing.

In this section, we show that every tree has a convex strongly monotone drawing. 
Moreover, if the tree is irreducible, then the drawing is strictly convex.
We use the result on trees to prove that every 
outerplanar graphs admits a strongly monotone drawing.

\begin{theorem}\label{theorem:tree}
Every tree has a convex strongly monotone drawing.
If the tree is irreducible, then the drawing is strictly convex.
\end{theorem}

\begin{proof}
  We actually prove something stronger, namely, that any tree $T$ has
  a drawing $\Gamma$ with the following properties:
\begin{enumerate}[label=(I\arabic*)]
\item\label{enum:prop-leaf}
  Every leaf of $T$ is placed on a corner of the convex hull of the vertices in $\Gamma$.
\item\label{enum:prop-ccw} If $a_1,\ldots,a_\ell$ is the
  counterclockwise order of the leaves on the convex hull, then for
  $i=1,\ldots,\ell$ the vectors
  $(\vec{a_ia_{i-1}})^\perp$, $\vec{p_{i},a_i}$,
  $(\vec{a_{i+1}a_i})^\perp$ appear in counterclockwise radial order,
  where $p_{i}$ denotes the unique vertex adjacent to~$a_i$.
\item\label{enum:prop-convex}
  The angle between two consecutive edges incident to a vertex $v\in
  V(T)$ is at most $\pi$ and is equal to $\pi$ only when $v$ has
  degree two.
\item\label{enum:prop-mono}
  $\Gamma$ is strongly monotone.
\end{enumerate}

Let $T$ be a tree on at least 3~vertices, rooted at some vertex~$v_0$
with degree at least~2.  We inductively produce a drawing of~$T$.  We
begin with placing the root $v_0$ at any point in the plane and the
children $u_1,\ldots,u_k$ of $v_0$ at the corners of a regular $k$-gon
with center~$v_0$.  The resulting drawing clearly fulfills the four
desired properties.

Let $T^-$ be a subtree of $T$ and let $\Gamma^-$ be a drawing of $T^-$
that fulfills the
properties~\ref{enum:prop-leaf}--\ref{enum:prop-mono}.  Let $a_i$ be a
leaf of $T^-$ and $u_1,\dots, u_k$ be the children of $a_i$ in
$T$. Let $T^+$ denote the subtree of $T$ induced by $V(T^-)\cup
\{u_1,\dots,u_k\}$.  In the inductive step, we explain how to extend
the drawing $\Gamma^-$ of $T^-$ to a drawing $\Gamma^+$ of $T^+$ such
that it fulfills the
properties~\ref{enum:prop-leaf}--\ref{enum:prop-mono}.

\begin{figure}[tb]
  \centering
  \subfloat[\label{fig:tree_Tk_1}]{
    \centering
    \includegraphics{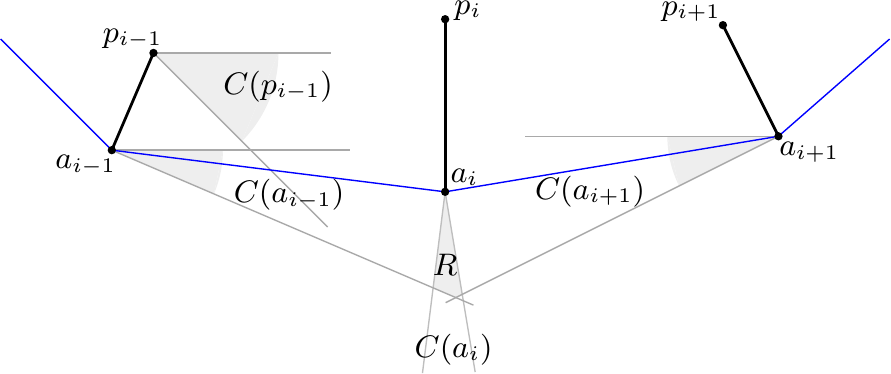}}
  \hfill
  \subfloat[\label{fig:tree_Tk_2}]{
    \centering
    \includegraphics{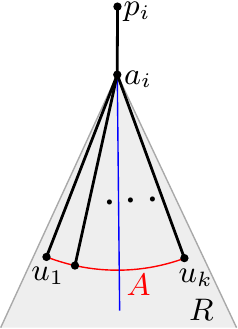}}
  \caption{
    (a)~The region~$R$ 
    which is used  for placing  all the children of vertex~$a_i$. 
    The boundary of the convex hull is drawn blue.
    (b)~Placement of the children $u_1,\ldots,u_k$ on the arc~$A \subset R$. 
    The prolongation~$h_{p_i,a_i}$ is drawn blue, the arc $A$ is drawn red.
    }
  \label{fig:tree_Tk_12}
\end{figure}

We first define a region $R$ which is appropriate for the placement of
$u_1,\dots,u_k$; see Figure~\ref{fig:tree_Tk_1} for an illustration.
Let $C(a_i)$ be the open cone containing all points $x$ such that the
vectors $(\vec{a_ia_{i-1}})^\perp$, $\vec{a_ix}$, and
$(\vec{a_{i+1}a_i})^\perp$ are ordered counterclockwise.
From property~\ref{enum:prop-ccw}, it follows
that $C(a_i)$ contains the \emph{prolongation} $h_{p_i,a_i}$ of
$\vec{p_ia_i}$, i.e., the ray that starts with $\vec{p_ia_i}$ and
extends across $a_i$.  For every vertex $y\neq a_i$ of~$T^-$, let
$C(y)$ be the open cone consisting of all points $p$ such that the
path from $y$ to $a_i$ in~$T^-$ is strictly monotone with respect
to~$\vec{yp}$.  Since the drawing $\Gamma^-$ is strongly monotone in a
strict sense,~$C(y)$ contains an open disk centered at $a_i$.  We
define the region $R$ to be the intersection of all these cones, i.e.,
$R=\cap_{y\in V(T^-)}{C(y)}$.  The intersection of the cones
$\{C(y)\mid y\in V(T^-)\setminus\{a_i\}\}$ contains an open disk
centered at $a_i$. The intersection of this disk with $C(a_i)$ yields
an open `pizza slice' contained in $R$. In particular, $R$ is
non-empty.

\begin{figure}[tb]
  \centering
  \subfloat[\label{fig:tree_Tk_4}]{
    \centering
    \includegraphics{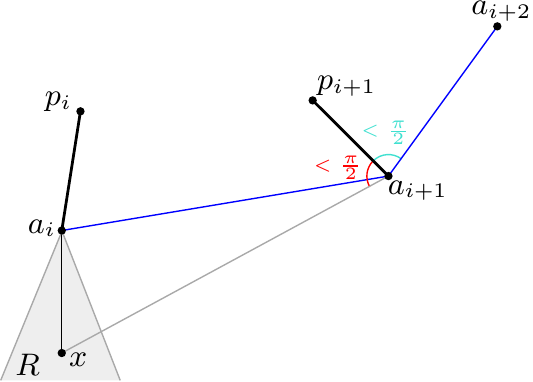}}
  \hfill
  \subfloat[\label{fig:tree_Tk_3}]{
    \centering
    \includegraphics{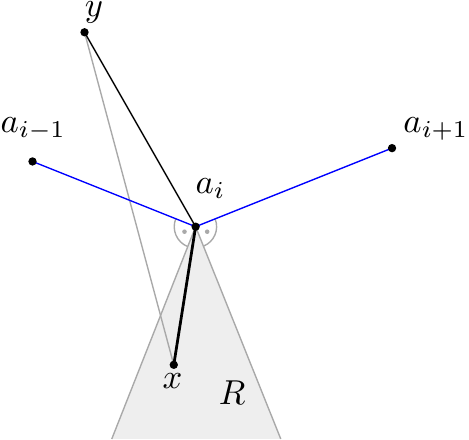}}
  \caption{
    (a)~An illustration for the proof of  property~\ref{enum:prop-leaf} and property~\ref{enum:prop-ccw}.  (b)~An illustration of the case where $y\in V(T^-)$ and $x \in \{u_1,\dots,u_k\}$.
    }
  \label{fig:tree_Tk}
\end{figure}

Since $R$ is an open convex set, we can construct a circular arc~$A$
in $R$ with center~$a_i$ that contains points on both sides of the
prolongation $h_{p_i,a_i}$ of $\vec{p_ia_i}$; see
Figure~\ref{fig:tree_Tk_2}.  We place the vertices $u_1,\dots,u_k$ on
the arc $A$ such that $\angle{p_i a_i u_1}=\angle{u_k a_i p_i}$.  This
placement implies that in case $a_i$ has degree $2$, $\angle {p_ia_i
  u_1}=\angle{u_k a_i p_i}=\pi$, and otherwise all the angles $\angle
{p_ia_i u_1}$, $\angle{u_k a_i p_i}$, $\angle{u_j a_i u_{j+1}}$, for
$j=1,\ldots,k-1$, are all less than $\pi$. This ensures
property~\ref{enum:prop-convex}.

Next, we prove that the drawing~$\Gamma^+$ of $T^+$ fulfills
property~\ref{enum:prop-leaf}.  We first show that $a_{i-1}$ and
$a_{i+1}$ lie on the convex hull of $\Gamma^+$; see
Figure~\ref{fig:tree_Tk_4}.  Consider the path from $a_{i+1}$ to~$a_i$
in~$T^-$, and let $x$ be a point in $R$. By definition of $R$, this
path is monotone (in a strict sense) with respect to $\vec{a_ix}$;
therefore, $\angle{p_{i+1}a_{i+1}x} < \pi/2$.  Considering the
strictly monotone path from $a_{i+2}$ to~$a_{i+1}$ in~$T^-$ we obtain
that $\angle{a_{i+2} a_{i+1} p_{i+1}} < \pi/2$.  The two inequalities
above sum up to $\angle{xa_{i+1}a_{i+2}} <\pi$ which means that
$a_{i+1}$ lies on the convex hull of $\Gamma^+$.  Analogously, we
obtain that $a_{i-1}$ lies on the convex hull of $\Gamma^+$.

Notice that at least one of $u_1,\ldots,u_k$ lies on the convex hull of $\Gamma^+$
since they are placed outside of the convex hull of $\Gamma^-$. On the other 
hand, the construction of the circular arc $A$ on which they are placed ensures 
that all of them lie on the convex hull of $\Gamma^+$. 

For property~\ref{enum:prop-ccw}, observe that $\angle{x a_i a_{i+1}}
> \pi/2$ holds for every $x \in C(a_i)$ (see
Figure~\ref{fig:tree_Tk_4}), and therefore $\angle{a_{i+1} x a_i} <
\pi/2$, as these two angles lie in the triangle $\triangle
xa_ia_{i+1}$. The last inequality implies property~\ref{enum:prop-ccw}
for $\Gamma^+$.

Finally, we show that property~\ref{enum:prop-mono} holds, i.e., that
$\Gamma^+$ is a strongly monotone drawing.  Consider $x,y \in V(T^+)$,
let $P_{xy}$ denote the path between $x$ and $y$ in $T^+$. We
distinguish the following three cases:
\begin{enumerate}
 \item If $x,y \in V(T^-)$, 
then the path  $P_{xy}$ is contained in~$T^-$. 
Since $\Gamma^-$ is a strongly monotone drawing by induction hypothesis, 
$P_{xy}$ is strongly monotone.
\item If $y \in V(T^-)$ and $x \in \{u_1,\dots,u_k\}$, then $P_{yx} =
  P_{ya_i} + (a_i,x)$; refer to Figure~\ref{fig:tree_Tk_3}.  The path
  $P_{ya_i}$ is monotone with respect to $\vec{yx}$ by construction
  because $x \in A \subset R \subset C(x)$. The definition of~$R$ also
  implies that $\angle a_{i-1}a_ix$ and $\angle a_{i+1}a_ix$ are
  greater than $\pi/2$.  Since~$y$ lies inside the convex hull of
  $\Gamma^-$, the smallest angle~$\angle ya_ix$ is also greater than
  $\pi/2$. Thus, $\angle{a_i x y}<\pi/2$ which implies that the
  vector~$\vec{xa_i}$ is monotone with respect to~$\vec{xy}$. We
  conclude that $P_{xy}$ is strongly monotone.
\item 
If $x,y \in \{u_1,\dots,u_k\}$, 
then the path $P_{xy}=(x,a_i)+(a_i,y)$ is strongly monotone 
since $x$ and $y$ are placed on the circular arc~$A$ centered at~$a_i$.
\end{enumerate}

\noindent
We have proven that each tree has a drawing that fulfills the four
properties~\ref{enum:prop-leaf}--\ref{enum:prop-mono}.
Property~\ref{enum:prop-ccw} implies that the prolongations of the
edges incident to the leaves do not intersect. This, together with
property~\ref{enum:prop-convex}, implies the convexity of the drawing
and strong convexity in case of an irreducible tree. This concludes
the proof of the theorem.
\end{proof}

\begin{theorem}
Every outerplanar graph has a convex strongly monotone drawing.
\end{theorem}
\begin{proof}
  Let $G$ be an outerplanar graph with at least 2 vertices.  For every
  vertex $v \in V$, we add two dummy vertices $v',v''$ and edges
  $(v,v'),~(v,v'')$.  By construction, the resulting graph~$H$ is
  outerplanar and does not contain vertices of degree~2.  
  Let~$\Gamma_{H}$ be an outerplanar drawing of~$H$. We will
  construct a convex strongly monotone drawing~$\Gamma_{H}'$ of~$H$
  with the same combinatorial embedding as $\Gamma_H$.

Let $T$ be an arbitrary spanning tree of~$H$.
By construction, no vertex in~$T$ has degree~2.
Thus, according to Theorem~\ref{theorem:tree}, 
$T$ admits a strongly monotone drawing~$\Gamma_{T}$ 
which is strictly convex and
which also preserves the order of the children 
for every vertex, i.e., the rotation system coincides with the one in~$\Gamma_{H}$.

Now, we insert all the missing edges. 
Recall that, by removing an edge from a planar drawing, the two adjacent faces are merged.
Since the drawing $\Gamma_{T}$ of~$T$ is strictly convex and 
since $\Gamma_{T}$ preserves the rotation system of $\Gamma_{H}$,
by inserting an edge~$e$ of the graph $H$ into $\Gamma_{T}$
one  strictly convex face is partitioned into two strictly convex faces.
Furthermore, the insertion of an edge does not destroy strong monotonicity.
We  re-insert all edges of $H$ iteratively.
The resulting drawing $\Gamma_{H}'$ of $H$  is a strictly convex and strongly monotone.

Finally, we remove all the dummy vertices and obtain a strongly monotone drawing of~$G$. 
Since $\Gamma_{H}'$ has the same combinatorial embedding as $\Gamma_{H}$, 
every dummy vertex lies in the outer face.
Hence, no internal face is affected by the removal of dummy vertices, and thus all interior faces remain strictly convex.
\end{proof}

\section{2-Trees}\label{sec:2trees}
In this section, we show how to construct a strongly monotone drawing for any 2-tree.
We begin by introducing some notation.
A \emph{drawing with bubbles} of a graph $G=(V,E)$ is a straight-line drawing of $G$
in the plane such that, for some $E'\subseteq E$, every edge~$e\in E'$ is associated 
with a circular region in the plane, called a \emph{bubble}~$B_e$; see 
Figure~\ref{fig:2:bubbles}. 
An {\em extension} of a drawing with bubbles is a straight-line drawing that is
obtained by taking some subset of edges with bubbles $E''\subseteq E'$ and
stacking one vertex on top of each edge $e\in E''$ into the corresponding
bubble $B_e$; see Figure~\ref{fig:2:extension}.
(Since every bubble is associated with a unique edge we often simply say
that a vertex is stacked into a bubble without mentioning the corresponding edge.)
We call a drawing with bubbles~$\Gamma$ strongly monotone 
if \emph{every} extension of~$\Gamma$ is strongly monotone.
Note that this implies that if a vertex~$w$ is stacked on top of edge~$e$ into
bubble~$B_e$, then there exists a strongly monotone path from~$w$ to any other
vertex in the drawing and, furthermore, there exists a strongly monotone path 
from~$w$ to any of the current bubbles, i.e., to any vertex that might be
stacked into another bubble.

Every 2-tree~$T=(V,E)$ can be constructed through the following iterative procedure:

\begin{enumerate}[label=(\arabic*)]
\item\label{enum:2tree-edge} We start with one edge and tag it as
  \emph{active}.  During the entire procedure, every present edge is
  tagged either as active or \emph{inactive}.
\item\label{enum:2tree-active} As an iterative step we pick one active
  edge~$e$ and stack vertices $w_1,\ldots,w_k$ on top of this edge for
  some $k \ge 0$ (we note that $k$ might equal $0$).  Edge $e$ is then
  tagged as inactive and all new edges incident to the stacked
  vertices $w_1,\ldots,w_k$ are tagged as active.
\item\label{enum:2tree-repeat} If there are active edges remaining,
  repeat Step~\ref{enum:2tree-active}.
\end{enumerate} 
Observe that Step~\ref{enum:2tree-active} is performed exactly once
per edge and that an according decomposition for~$T$ can always be
found by the definition of 2-trees.

\begin{figure}[tb]
  \centering
  \subfloat[\label{fig:2:bubbles}]{
    \centering
    \includegraphics[page=1]{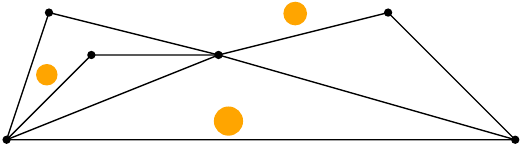}}
  \hfill
  \subfloat[\label{fig:2:extension}]{
    \centering
    \includegraphics[page=2]{extension}}
  \caption{(a) A drawing of a 2-tree with bubbles (orange) and (b) an extension of the
    drawing.}
\end{figure}

We construct a strongly monotone drawing of $T$ 
by geometrically implementing the iterative procedure described above,
so that after every step of the algorithm the present part of the graph is realized
as a drawing with bubbles. 
We use the following additional geometrical invariant:
\begin{enumerate}[label=(C)]
\item\label{item:condition}
    After each step of the algorithm every active edge comes with a bubble
    and the drawing with bubbles is strongly monotone.
    Additionaly, for an edge $e=(uv)$ with bubble $B_e$ 
    for each point~$w\in B_e$, 
    the angle~$\angle(\vec{uw},\vec{wv})$ is obtuse. 
\end{enumerate}

In Step~\ref{enum:2tree-edge}, we arbitrarily draw the edge~$e_0$ in the plane. 
Clearly, it is possible to define a bubble for~$e_0$ that only allows obtuse
angles. In Step~\ref{enum:2tree-active}, we place the vertices~$w_1,\ldots,w_k$ 
over an edge~$e=(u,v)$ as follows. The fact that stacking a vertex into~$B_e$ 
gives an obtuse angle allows us to place the to-be stacked 
vertices~$w_1,\dots w_k$ in~$B_e$ on a circular arc around~$u$ such that, for 
any~$1\le i,j\le k$, there exists a strongly monotone path between~$w_i$ 
and~$w_j$; see Figure~\ref{fig:add_new_points}. Due to 
condition~\ref{item:condition}, there also exists a strongly monotone path 
between any of the newly stacked vertices and any vertex of an extension of  
the previous drawing with bubbles. Hence, after removing the bubble~$B_e$, the resulting drawing is
a strongly monotone drawing with bubbles.

In order to maintain condition~\ref{item:condition}, it remains to describe how 
to define the bubbles for the new active edges incident to the stacked vertices. 
For this purpose, we state the following Lemma~\ref{lm:bubbles}, which enables 
us to define the two bubbles for the edges incident to any degree-2 vertex with
an obtuse angle. The Lemma is then iteratively applied to the 
vertices~$w_1,\dots ,w_k$
and after every usage of the Lemma the produced drawing with bubbles is strongly monotone.
This iterative approach is used to ensure that, when 
defining bubbles for some vertex~$w_i$, the previously added bubbles 
for~$w_1,\dots,w_{i-1}$ are taken into account.

\begin{figure}[tb]
  \centering
  \subfloat[Stacking vertices into a bubble\label{fig:add_new_points}]{
    \centering
    \includegraphics[scale=1]{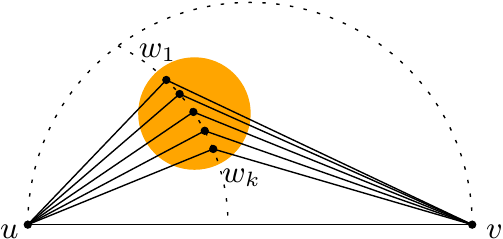}}
  \hfill
  \subfloat[The empty neighbourhood~$\mathcal N$ (dotted)\label{fig:2:planar}]{
    \centering
    \hspace{2mm}
    \includegraphics[scale=1]{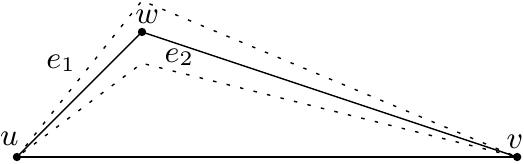}}
  \caption{Illustrations for the drawing approach for strongly monotone 2-trees.}
\end{figure}

\begin{lemma}
\label{lm:bubbles}
Let~$\Gamma$ be a strongly monotone drawing with bubbles and let~$w$ be a vertex
of degree 2 with an obtuse angle such that the two incident edges~$e_1=(u,w)$ 
and~$e_2=(v,w)$ have no bubbles. Then, there exist bubbles~$B_{e_1}$ 
and~$B_{e_2}$ for edges $e_1$ and $e_2$ respectively that only allow obtuse angles such 
that $\Gamma$ remains strongly monotone with bubbles if we add~$B_{e_1}$ 
and~$B_{e_2}$.
\end{lemma}

\begin{proof}
  We begin by describing how we determine the size and location of the
  new bubbles.  Since~$\Gamma$ is planar, there exists a
  neighborhood~$\mathcal N$ of~$w$, $e_1$ and~$e_2$ that does not
  contain elements of any extension of~$\Gamma$; see
  Figure~\ref{fig:2:planar}.

  Furthermore, consider any extension~$\Gamma'$ of~$\Gamma$. Since we
  consider monotonicity in a strict fashion, there exists a
  constant~$\alpha(\Gamma')>0$ such that, for any pair of
  vertices~$s_0,s_t$ of~$\Gamma'$ and for any strongly monotone
  path~$P=(s_0,\dots ,s_t)$ it holds that $\angle
  (\vec{s_0s_t},\vec{s_{i}s_{i+1}})<
  \pi/2-\alpha(\Gamma')$ for $i=0,\dots,t-1$. We refer to this property
  of~$P$ as being~\emph{$\alpha(\Gamma')$-safe} with respect
  to~$\vec{s_os_t}$.  A simple compactness argument shows that this
  safety parameter can be chosen simultaneously for all the extensions
  of $\Gamma$: there exists $\alpha>0$ such that for every extension
  $\Gamma'$ of $\Gamma$ for every two vertices~$s_0,s_t$ of~$\Gamma'$
  every strongly monotone path connecting these vertices is
  ~$\alpha$-safe with respect to~$\vec{s_os_t}$.  (This global
  constant can be chosen as $\alpha := \min_{\Gamma'}\alpha(\Gamma')$,
  where the minimum is taken over all the extensions $\Gamma'$ of
  $\Gamma$, and the minimum is strictly positive since the set of
  extensions is compact.)

  For the edge $e_1$, we define the bubble $B_{e_1}$ as the circle of
  radius $r$ with center  at the extension of the edge
  $e_2$ over $w$ with distance $\eps$ to $w$ as depicted in
  Figure~\ref{fig:add_new_bubble}.  In order to ensure the strong
  monotonicity, we choose $r$ and $\eps$ such that the following
  properties hold (these properties clearly hold as soon as $r$,
  $\eps$ and $r/\eps$ are small enough):

\begin{figure}[tb]
  \centering
  \subfloat[Semi circles around~$e_1$ and~$e_2$\label{fig:add_new_bubble}]{
    \centering
    \includegraphics[scale=1]{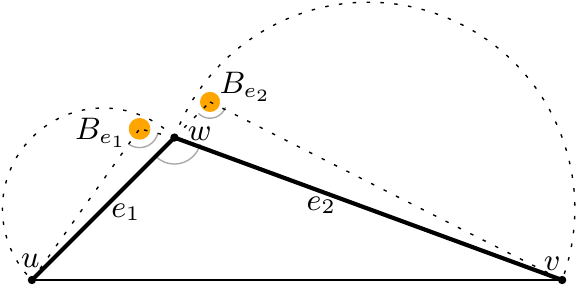}}
  \hfill
  \subfloat[Cones centered at~$u$ and~$w$\label{fig:add_new_numbers}]{
    \centering
    \includegraphics[scale=1]{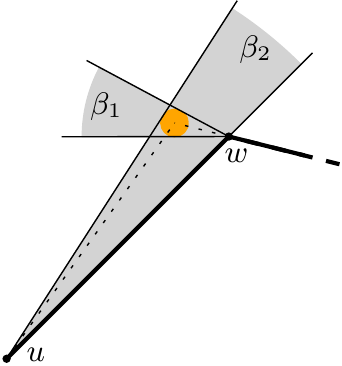}}
  \hfill  
  \subfloat[Cone centered at~$y$\label{fig:add_new_cone}]{
    \centering
    \hspace{3mm}
    \includegraphics[scale=1]{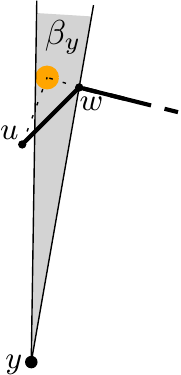}}
  \caption{Illustrations for the placement of the new bubbles~$B_{e_1}$ and~$B_{e_2}$}
\end{figure}

\begin{enumerate}[label=(\roman*)]
  \item \label{lemma_item_1}
    Bubble $B_{e_1}$ is located inside the  empty neighborhood $\mathcal N$. 
    Moreover, to preserve obtusity, $B_{e_1}$ needs to lie inside the semicircle 
    with edge $e_1$ as diameter, as depicted in 
    Figure~\ref{fig:add_new_bubble}.
  \item \label{lemma_item_2}
    Consider angles $\beta_1$ and $\beta_2$ as illustrated in 
    Figure~\ref{fig:add_new_numbers}. We require that both angles are smaller 
    than $\alpha/4$.
  \item \label{lemma_item_3}
    For any vertex~$y$ of any extension of~$\Gamma$, consider the angle $\beta_y$
    as illustrated in Figure~\ref{fig:add_new_cone}.
    We require that this angle is smaller than $\alpha/4$.
    That guarantees that for any 
    point~$x\in B_{e_1}$ it holds that $\angle (\vec{yw},\vec{yx})<\alpha/4$.
\end{enumerate}

We define the bubble $B_{e_2}$ for the edge $e_2$ analogously with
$B_{e_1}$.  Moreover, we can use the same pair of parameters $r$
and $\eps$ for $B_{e_1}$ and $B_{e_2}$.

For the strong monotonicity of the drawing $\Gamma$ with two new
bubbles $B_{e_1}$ and $B_{e_2}$ we have to show two conditions:
(1)~that from any vertex stacked into one of the new bubbles there
exists a strongly monotone path to any vertex~$y$ of any extension
of~$\Gamma$ and (2)~that there exists a strongly monotone path between
any vertex stacked into $B_{e_1}$ and any vertex stacked into
$B_{e_2}$.

Since we use the same pair of $r$ and $\eps$ for defining $B_{e_1}$ and $B_{e_2}$,
the condition~(2) clearly holds as soon as $r/\eps$ is small enough.
Thus we are left with ensuring that the condition (1) holds.

Consider the new bubble $B_{e_1}$, a point $x \in B_{e_1}$ and any vertex $y$
of any extension $\Gamma'$ of $\Gamma$.
Since the drawing $\Gamma'$ is strongly monotone, there exists a strongly 
monotone path~$P_{yw}$ in $\Gamma'$ between~$y$ and~$w$.
Since $w$ has only two incident edges in $\Gamma'$, the last edge of the path~$P_{yw}$
is either $e_1=(u,w)$ or $e_2=(v,w)$.
We distinguish between these two cases: in the first case we 
construct a path $P_{yx}$ from~$y$ to~$x$ by re-routing the last edge of 
$P_{yw}$ from $(u,w)$ to $(u,x)$ as illustrated in Figure~\ref{fig:add_new_paths_a};
in the second case we 
construct a path $P_{yx}$ by appending the edge $(w,x)$ to the end of
$P_{yw}$ as illustrated in Figure~\ref{fig:add_new_paths_b};

It remains to show that~$P_{yx}$ is strongly monotone.  First, observe
that~$P_{yw}$ is strongly monotone and~$\alpha$-safe.  By
property~\ref{lemma_item_2}, the final edges~$e_w$ of~$P_{yw}$
and~$e_x$ of~$P_{yx}$ satisfy $\angle(\vec{e_w},\vec{e_x})<\alpha/4$
and all other edges of these paths are identical.  Thus,~$P_{yx}$
is~$(3\alpha/4)$-safe with respect to~$\vec{yw}$.  By
Property~\ref{lemma_item_3} $\angle (\vec{yw},\vec{yx})<\alpha/4$ and,
therefore,~$P_{yx}$ is~$(\alpha/2)$-safe with respect to $\vec{yx}$
and thus in particular it is strongly monotone.

\begin{figure}[tb]
  \centering
  \subfloat[Rerouting in case the last edge of $P_{yw}$ is $e_1$.\label{fig:add_new_paths_a}]{
    \centering
    \includegraphics[scale=1]{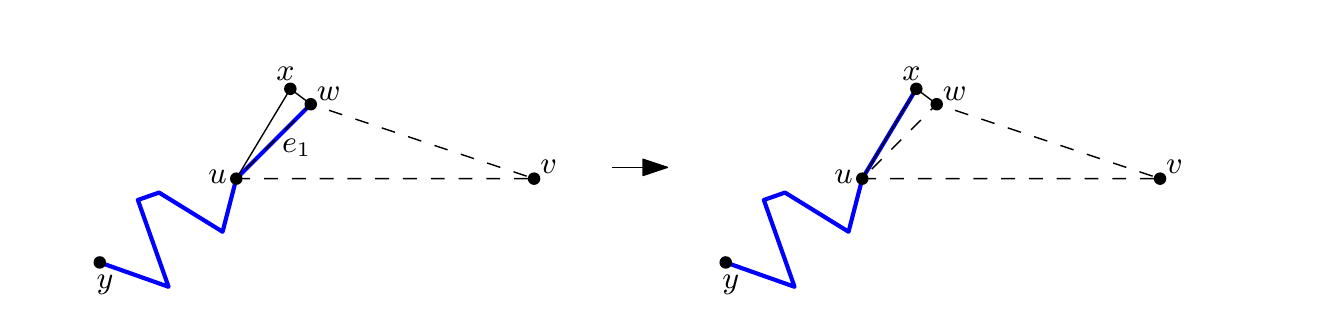}}
  \hfill
  \centering
  \subfloat[Rerouting in case the last edge of $P_{yw}$ is $e_2$.\label{fig:add_new_paths_b}]{
    \centering
    \includegraphics[scale=1]{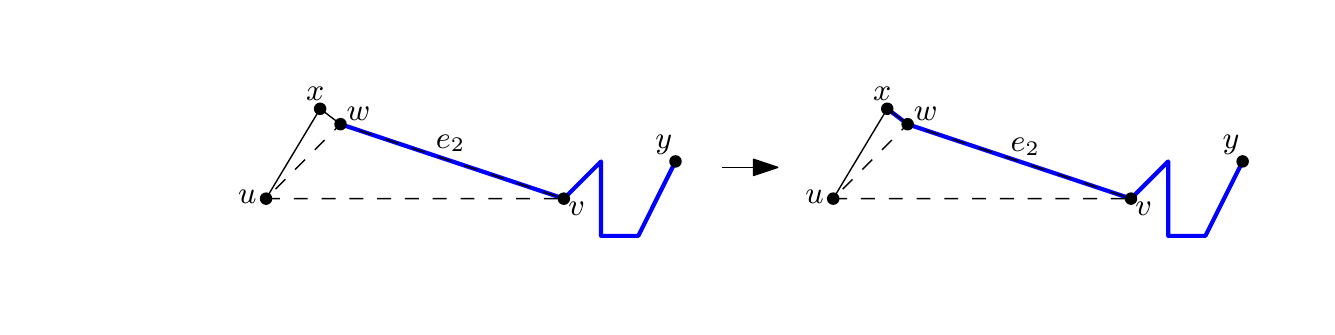}}
  \hfill
\caption{A strongly monotone path $P_{yw}$ from $y$ to $w$ is re-routed to $x$.
Two possible cases are distinguished: The last edge of $P_{yw}$ is either $e_1$ or $e_2$.
}
\label{fig:add_new_paths}
\end{figure}
The arguments for a vertex stacked on~$e_2$ into~$B_{e_2}$ are identical.
\end{proof}
Thus, we obtain the main result of this section:
\begin{theorem}\label{thm:2tree}
Every 2-tree admits a strongly monotone drawing.
\end{theorem}

\section{Conclusion}\label{sec:conclusion}
We have shown that any 3-connected planar graph, tree, outerplanar graph, and
2-tree admits a strongly monotone drawing. All our drawings require exponential 
area. For trees, this area bound has been proven to be required; however, it
remains open whether the other graph classes can be drawn in polynomial area.
Further, the question whether any 2-connected planar graph admits a strongly 
monotone drawing remains open. 
Last but not least, we could observe (using a computer-assisted search) that 
2-connected graphs with at most~9 vertices admit a strongly monotone drawing, 
while there is exactly one connected graph with~7 vertices that is the smallest 
graph not admitting a strongly monotone drawing; 
see Figure~\ref{fig:connected_7_vertex_graph_without_str_m_drawing}.

\begin{figure}[t]
  \centering
  \includegraphics{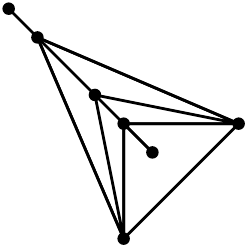}
  \caption{The unique connected 7-vertex graph without a strongly monotone drawing.}
  \label{fig:connected_7_vertex_graph_without_str_m_drawing}
\end{figure}

\bibliographystyle{plain}
\bibliography{abbrv,monotone}

\end{document}